\documentclass[11pt]{article}

\usepackage[utf8]{inputenc} % set input encoding (not needed with XeLaTeX)
\usepackage{latexsym}
\usepackage{amssymb}
\usepackage{mathrsfs}%mathscr
\usepackage[table,xcdraw]{xcolor}
\usepackage{amsmath}%包含amsthm，用于newtheorem
\usepackage{amsthm}% theoremstyle

\theoremstyle{definition}
\newtheorem{defi}{Definition}[section]
\newtheorem{rmk}{Remark}[section]
\newtheorem{prop}{Proposition}[section]

\newtheorem{thm}{Theorem}[section]
\newtheorem{eg}{Example}[section]
\newtheorem{ax}{Axiom}

\newtheorem{lemma}{Lemma}[section]

\newtheoremstyle{mystyle}
{3pt}% Space above 
{3pt}% Space below
{\upshape}% Body font
{}% Indent amount
{}% Theorem head font
{:}% Punctuation after theorem head
{.5em}% Space after theorem head
{}% Theorem head spec (can be left empty, meaning ‘normal’ )
\theoremstyle{mystyle}

\usepackage{pythonhighlight}
\usepackage{geometry} % to change the page dimensions
\usepackage{graphicx} % support the \includegraphics command and options
\usepackage{booktabs} % for much better looking tables
\usepackage{array} % for better arrays (eg matrices) in maths
\usepackage{paralist} % very flexible & customisable lists (eg. enumerate/itemize, etc.)
\usepackage{verbatim} % adds environment for commenting out blocks of text & for better verbatim
\usepackage{subfig} % make it possible to include more than one captioned figure/table in a single float
\usepackage{fancyhdr} % This should be set AFTER setting up the page geometry
\usepackage{sectsty}
\allsectionsfont{\sffamily\mdseries\upshape} % (See the fntguide.pdf for font help)
\usepackage[nottoc,notlof,notlot]{tocbibind} % Put the bibliography in the ToC
\usepackage[titles,subfigure]{tocloft} % Alter the style of the Table of Contents

\usepackage[colorlinks,linkcolor=blue]{hyperref}
\usepackage[all]{xy}
\usepackage{cancel}

\title{Representation theory in the construction of free quantum field}
\author{Feng Zixuan}

\begin{document}
\maketitle
\tableofcontents

\clearpage

\begin{abstract}
    This is mainly a lecture note taken by myself following Weinberg's book \cite{weinberg2002quantum}, but also contains some corrections to the abuse of mathematical treatment. This article discusses projective unitary representations of Poincare group on the single particle space, multi-particle space also known as the Fock space, creation and annilation operators, construction of free quantum fields and the general relation between spin of state and spin of field. Both massive and massless cases are considered. CPT is not considered. The first section briefly reviews the basics of representation theory. This article further points out some of the wrong treatment of mathematics in the book \cite{weinberg2002quantum}, and reformulates them, including: Wigner's classification needs to be pass to the universal cover via Bargmann's theorem, there is no projective representation of Poincare group on Fock space in general, the Lorentz transformation of fields need to be formulated with representations of the universal covers, Dirac representation is not a linear representation of the Lorentz group. This article also discusses the physical meaning of the state representation and its relation with Schrodinger equation, compare its difference with state representation, and the reason that equations of relativistic quantum mechanics should be understood as field equations rather than a wave function equations. 
\end{abstract}

\clearpage

\section{Representation theory of the Lorentz group}

We denote by $SO^+(1,3)$ the identity component of the homogenous Lorentz group, and 
$\mathbb{R}^{1,3} \rtimes SO^+(1,3)$ the identity component of the Poincare group. As we do not consider the whole group containing $P$ and $T$, they are abbreviated Lorentz group and Poincare group.

\subsection{$sl_2(\mathbb{C})$}

This well-known result can be found in many textbooks. For example, see \cite{humphreys2012introduction}.

Recall we have a complex basis for $sl_2(\mathbb{C})$:
\begin{equation}
    H = \begin{pmatrix}1&0\\0&-1 \end{pmatrix}
    \quad\quad
    E = \begin{pmatrix}0&1\\0&0 \end{pmatrix}
    \quad\quad
    F = \begin{pmatrix}0&0\\1&0 \end{pmatrix}
\end{equation}
and all the irreducible representations are
\begin{equation}
    \begin{aligned}
    V_\lambda &= span\{v_0, \cdots, v_\lambda \}\\
    Hv_k &= (\lambda-2k)v_k\\
    F v_k &= v_{k+1}\\
    E v_k &= k(\lambda - k+1)v_{k-1}\\
    v_{-1} &= v_{\lambda+1} \equiv 0
    \end{aligned}
\end{equation}
Written in matrix form,
\begin{equation}\label{sl_2}
    H \mapsto 
    \begin{pmatrix}
    \lambda & & &\\
    & \lambda -2 & &\\
    & &  \cdots &\\
    &&& -\lambda
    \end{pmatrix}
    \quad\quad
    F \mapsto
    \begin{pmatrix}
    0 & & &\\
    1 & 0 & &\\
    & 1 &  \cdots &\\
    & & 1 & 0
    \end{pmatrix}
    \quad\quad
    E \mapsto
    \begin{pmatrix}
    0 & i(\lambda -i+1) & &\\
     & 0 & &\\
    &  &  \cdots &\\
    & &  & 0
    \end{pmatrix}
    ,i=1,\cdots,\lambda
\end{equation}

\subsection{Angular momentum}

Angular momentum is a representation of $su(2)$. It relates with the previous example by
\begin{equation}
    span_\mathbb{R} \langle iX,iY,iZ \rangle = su(2) \hookrightarrow su(2)_\mathbb{C} 
    = sl(2, \mathbb{C}) = span_\mathbb{R} \langle X,Y,Z, iX,iY,iZ\rangle
\end{equation}

\begin{rmk}
    What seems unfamiliar is that the basis of $su(2)$ is i times Pauli matrices, instead of themselves. This is the departure of mathematical language and physical language. Here I use the mathematical language. 
\end{rmk}

When passing to angular momentum language, we have
\begin{equation}
    \begin{aligned}
        J_z &= Z/2 = H/2\\
        J_x &= X/2 = (E+F)/2\\
        J_y &= Y/2 = (E-F)/2i\\
        J_+ &= E\\
        J_- &= F
    \end{aligned}
\end{equation}
and $\lambda = 2j$. Rearrange the indices $0\to \lambda$ as $+j\to -j$, and do similitude transformation in order to make $E$ and $F$ symmetric, we can rewrite \eqref{sl_2} as
\begin{equation}\label{angular momentum}
\begin{aligned}
    (J_z)_{\sigma',\sigma} &= \sigma \delta_{\sigma',\sigma}\\
    (J_x \pm iJ_y)_{\sigma',\sigma} &= (J_\pm)_{\sigma',\sigma} = \delta_{\sigma',\sigma\pm 1} \sqrt{(j\mp \sigma)(j\pm \sigma +1)}
\end{aligned}
\end{equation}
These are all the finite dimensional irreducible representations of $su(2)$.

\subsection{Representation of the Lorentz algebra}

We refer to \cite{enwiki:1123050261}. 

The standard procedure for finding finite dimensional representations for $so(1,3)$ is to set
\begin{equation}
\mathbf{A}=\frac{\mathbf{J}+i \mathbf{K}}{2}, \quad \mathbf{B}=\frac{\mathbf{J}-i \mathbf{K}}{2} .
\end{equation}
and thus we can prove
\begin{equation}
[A_i, A_j]=i \varepsilon_{i j k} A_k, \quad[B_i, B_j]=i \varepsilon_{i j k} B_k, \quad[A_i, B_j]=0,
\end{equation}

What we have done in mathematical language is to complexify the Lorentz Lie algebra and discover it to decompose into two $su(2)_\mathbb{C}$:
\begin{equation}
    so(1,3) \hookrightarrow so(1,3)_\mathbb{C} = su(2)_\mathbb{C} \oplus su(2)_\mathbb{C} = sl(2,\mathbb{C}) \oplus sl(2,\mathbb{C})
\end{equation}

\begin{rmk}
    The $A_i$'s and $B_i$'s do not lie in the Lorentz algebra, because they carry an imaginary coefficient. The expression
    \begin{equation}
    so(1,3)  = su(2) \oplus su(2)
    \end{equation}
    is actually wrong!
\end{rmk}

So the representation of the Lorentz algebra is equivalent to two representations of the angular momentum algebra $su(2)$. This is due to the following non-trivial fact. 
\begin{lemma}\footnote{I haven't found any reference to this proposition.}
    Given two finite dimensional Lie algebras $g_1,g_2$, then any finite dimensional irreducible rep of $g_1\oplus g_2$ must come from the tensor product of two f.d irreducible reps of the two Lie algebras, i.e.
    \begin{equation}
        \pi(v_1,v_2) = \pi_1(v_1) \otimes I + I\otimes \pi_2(v_2)
    \end{equation}
    We write $\pi = \pi_1 \boxtimes \pi_2$,
    \footnote{This notation is by myself. }
    differing from $\otimes$ which we use for tensor product of two reps for the same Lie algebra. 
\end{lemma}

Denote by $A$ and $B$ the spins of the two representations of $su(2)$ or $su(2)_\mathbb{C}$.

By definition we have
\begin{equation}
    \mathbf{J} = \mathbf{A} + \mathbf{B}, \quad\quad \mathbf{K} = -i(\mathbf{A} - \mathbf{B})
\end{equation}
The tensor product of representations of Lie algebra:
\begin{equation}
\begin{aligned}
& \pi_{(A, B)}(J_i)=J_i^{(A)} \otimes 1_{(2 B+1)}+1_{(2 A+1)} \otimes J_i^{(B)} \\
& \pi_{(A, B)}(K_i)=-i(J_i^{(A)} \otimes 1_{(2 B+1)} - 1_{(2 A+1)} \otimes J_i^{(B)}),
\end{aligned}
\end{equation}
where $1_n$ is the $n$-dimensional unit matrix and
\begin{equation}
\mathbf{J}^{(j)}=(J_1^{(j)}, J_2^{(j)}, J_3^{(j)})
\end{equation}
are the $(2 j+1)$-dimensional irreducible representations of $su(2)$.

Recall from the last subsection that in matrix expression,
\begin{equation}
\begin{aligned}
& (J_1^{(j)})_{a' a}=\frac{1}{2}(\sqrt{(j-a)(j+a+1)} \delta_{a', a+1}+\sqrt{(j+a)(j-a+1)} \delta_{a', a-1}) \\
& (J_2^{(j)})_{a' a}=\frac{1}{2 i}(\sqrt{(j-a)(j+a+1)} \delta_{a', a+1}-\sqrt{(j+a)(j-a+1)} \delta_{a', a-1}) \\
& (J_3^{(j)})_{a' a}=a \delta_{a', a}
\end{aligned}
\end{equation}
With $-A \leq a, a' \leq A,-B \leq b, b' \leq B$, the matrix elements of the representation is
\begin{equation}
\begin{gathered}
(\pi_{(A, B)}(J_i))_{a' b', a b}=\delta_{b'b}(J_i^{(A)})_{a' a}+\delta_{a' a}(J_i^{(B)})_{b' b} \\
(\pi_{(A, B)}(K_i))_{a' b' ,a b}=-i(\delta_{b' b}(J_i^{(A)})_{a' a} - \delta_{a' a}(J_i^{(B)})_{b' b})
\end{gathered}
\end{equation}

\begin{lemma}
\footnote{Though it's basic, this is by myself. I haven't find the reference. }
    We have an obvious fact:
    \begin{equation*}
        \pi_A \boxtimes \pi_B |_{span\mathbf{J}} \cong \pi_A \otimes \pi_B
    \end{equation*}
    as representations of $su(2)$. 
\end{lemma}
\begin{proof}
    This comes from the following general relation between $\boxtimes$ and $\otimes$. Let $\pi_1,\pi_2$ be two reps of a same Lie algebra $g$. Then under diagonal action $\iota:g\hookrightarrow g\oplus g$, we have $\iota^*(\pi_1 \boxtimes \pi_2) = \pi_1 \otimes \pi_2$.
\end{proof}

\subsection{Lorentz Lie group}

In physics sometimes we need projective representations, instead of usual ones. We refer to \cite{enwiki:1123050261}. 

\begin{lemma}(Corollary of Lie's second theorem)\\
\footnote{See \cite{warner1983foundations}. }
    If $\tilde{G}$ is a connected, simply connected Lie group and $g$ its Lie algebra. Then the finite dimensional rep of $g$ 1-1 corresponds to f.d. rep of $\tilde{G}$. 
\end{lemma}

This is the power of simply connectivity. In the general case it is not true, but we can pass to its universal covering space. 

\begin{lemma}
    Suppose $G$ is a connected Lie group and $\pi: \tilde{G} \to G$ its universal cover. Then every f.d. irr rep of $\tilde{G}$ induces a f.d. irr projective rep of $G$. 
\end{lemma}
\begin{proof}
    First suppose that $Ker\pi$ lies in the center of $\tilde{G}$. Given a f.d. irr rep $\rho: \tilde{G}\to GL(V)$, we define a f.d. irr projective rep of $G$ as follows. For any $g\in G$, choose a preimage $\tilde{g}\in\tilde{G}$, and set $\rho*(g)=r\rho(\tilde{g})$, where $r:GL(V)\to PGL(V)$. This is well defined since for two preimages $\tilde{g}_1, \tilde{g}_2$, we have $\tilde{g}_1 \tilde{g}_2^{-1}$ lies in the center of $\tilde{G}$ by assumption. So by Schur's lemma, we have $\rho(\tilde{g}_1)\rho(\tilde{g}_2^{-1})$ is a scalar. So $r\rho(\tilde{g}_1)=r\rho(\tilde{g}_2)$. 
    $$
    \xymatrix{
    \tilde{G} \ar[d]_\pi \ar[r]^\rho  &  GL(V)\ar[d]_r \\
    G  \ar@{-->}[r] &  PGL(V)
    }
    $$
    Notice that the descending can be an ordinary rep of $G$ iff $Ker\pi$ acts on $V$ trivially, and a projective rep of $G$ iff $Ker\pi$ acts as scalars. 
    
    Now we can show that it is always true that $Ker\pi$ lies in the center of $\tilde{G}$. Clearly $N=Ker\pi$ is a discrete normal subgroup. Suppose $n\in N, k\in\tilde{G}$. By connectivity, find a path $k(t)$ in $\tilde{G}$ going from $e$ to $k$. By normality, $k(t)nk(t)^{-1}$ is a path lying entirely in $N$ from $n$ to $knk^{-1}$. By discreteness, we are done. 
\end{proof}

The converse is also true, when restricted to finite dimensional case. 

Sadly, the converse is not true in infinite dimensional case: projective rep of $G$ not necessarily induce an ordinary rep of $\tilde{G}$. Bargmann's theorem gives a criterion under which every irr projective rep of $G$ arises in this way. 

\begin{lemma}(Bargmann's theorem)\label{Bargmann}
    \footnote{See \cite{poon2022projective}. }\\
    If the two-dimensional Lie algebra cohomology $H^2(g,\mathbb{R})$ is trivial, then every projective unitary rep of $G$ arises to an ordinary unitary rep of its universal cover. 
\end{lemma}
\begin{eg}
    The result does apply to semisimple groups (e.g., $SO(3)$) by Whitehead's lemma, and Lorentz group and the Poincare group. This is important for Wigner's classification of the projective unitary representations of the Poincare group. This is the content of section 2.1 in this article. 
\end{eg}

Let us find the universal cover of the Lorentz group and Poincare group. 

\begin{lemma}
    The universal cover of $SO^+(1,3)$ is $SL(2,\mathbb{C})$. The kernel is $\{\pm I\}$. 
\end{lemma}
\begin{proof}
    Let $SL(2,\mathbb{C})$ act on the set of all Hermitian $2\times 2$ matrices $h$ by
    \begin{equation*}
        P(A): h\to h, X\mapsto A^\dagger XA
    \end{equation*}
    Thus $P:SL(2,\mathbb{C}) \to GL(h)$ is a group homo. The kernel is $\pm I$ as taking $X=I$ in $X=A^\dagger XA$ means $A^\dagger = A^{-1}$. Thus $AX=XA$ for all $X\in h$, so $A=\pm I$ since $det A=1$. 

    We then identify $h$ with $\mathbb{R}^{1,3}$ by 
    \begin{equation*}
        X = 
        \begin{pmatrix}
            t+z & x+iy\\
            x-iy & t-z
        \end{pmatrix}
        \mapsto (t,x,y,z)
    \end{equation*}
    Thus $P$ becomes $p:SL(2,\mathbb{C}) \to GL(\mathbb{R}^{1,3})$. 

    Step 1. The image lies in the Lorentz group. 
    This is because the identification takes $det$ to $-||\cdot||^2$, and $SL(2,\mathbb{C})$ preserves determinant. 

    Step 2.  The mapping is smooth. So the image lies in $SO^+(1,3)$ as $SL(2,\mathbb{C})$ is connected. 

    Step 3. The kernel is $\pm I$. We have seen that. 

    Step 4. $SO^+(1,3)$ and $SL(2,\mathbb{C})$ both have dimension 6. So the image is an open subgroup in $SO^+(1,3)$, and thus they are equal. 

    Step 5. $SL(2,\mathbb{C})$ is simply connected. This is due to polar decomposition. As proved in the next lemma. 
\end{proof}

\begin{lemma}
    $SL(2,\mathbb{C})$ is simply connected. 
\end{lemma}
\begin{proof}
    By the polar decomposition theorem, any matrix $\lambda \in \mathrm{SL}(2, \mathbb{C})$ may be uniquely expressed as ${ }^{[71]}$
    \begin{equation}
    \lambda=u e^h
    \end{equation}
    where $u$ is unitary with determinant one, hence in $\mathrm{SU}(2)$, and $h$ is Hermitian with trace zero. The trace and determinant conditions imply: 
    \begin{equation}
    \begin{aligned}
    h&=\begin{pmatrix}
    c & a-i b \\
    a+i b & -c
    \end{pmatrix}, \quad\quad (a, b, c) \in \mathbb{R}^3 \\
    u&=\begin{pmatrix}
    d+i e & f+i g \\
    -f+i g & d-i e
    \end{pmatrix} \quad\quad (d, e, f, g) \in \mathbb{R}^4 \text { subject to } d^2+e^2+f^2+g^2=1
    \end{aligned}
    \end{equation}
    The manifestly continuous one-to-one map is a homeomorphism with continuous inverse given by (the locus of $u$ is identified with $\mathbb{S}^3 \subset \mathbb{R}^4$ )
    \begin{equation}
    \{\begin{array}{l}
    \mathbb{R}^3 \times \mathbb{S}^3 \to \mathrm{SL}(2, \mathbb{C}) \\
    (r, s) \mapsto u(s) e^{h(r)}
    \end{array}.
    \end{equation}
    explicitly exhibiting that $\mathrm{SL}(2, \mathbb{C})$ is simply connected.
\end{proof}

\begin{lemma}
        The universal cover of the Poincare group $\mathbb{R}^{1,3}\rtimes SO^+(1,3)$ is $\mathbb{R}^{1,3}\rtimes SL(2,\mathbb{C})$. The action of $SL(2,\mathbb{C})$ on $\mathbb{R}^{1,3}$ is $\Lambda,a\mapsto \lambda(\Lambda)a$. 
\end{lemma}

So when considering projective rep of the Lorentz group or Poincare group, we treat ordinary one of $SL(2,\mathbb{C})$ instead. We next study the behaviour of the small group that defined in section 2.1. 

\begin{lemma}
    Let $p:SL(2,\mathbb{C}) \to SO^+(1,3)$ be the covering map. Then the preimage of $SO(3)$ is $SU(2)$. 
\end{lemma}
\begin{proof}
    $p(A)\in SO(3)$ iff $p(A)$ acting on $\mathbb{R}^{1,3}$ fixes $(1,0,0,0)$, iff $P(A)$ acting on $h$ fixes $I$, iff $A^\dagger A = I$, iff $A\in SU(2)$. 
\end{proof}
\begin{rmk}
    $SU(2)$ is simply connected, so it gives the universal cover of $SO(3)$. But in general, the preimage of a subgroup of $SO^+(1,3)$ may not be simply connected. 
\end{rmk}

\begin{lemma}\footnote{This conclusion is mentioned partly in Wiki, but I haven't seen any whole statement or proof. I finished this alone.}
    The representation of $SL(2,\mathbb{C})$ arising from the rep $\bigoplus_{i=1}^n \pi_{A_iB_i}$ of its Lie algebra can be descended to an ordinary rep of $SO^+(1,3)$ iff $A_i+B_i$
    is an integer for all $i$, and can be descended to a projective rep of $SO^+(1,3)$ iff $A_i+B_i$
    is half an odd for all $i$. 
\end{lemma}
\begin{proof}
    First consider $i=1$. As 
    $$
    \xymatrix{
    \tilde{G} \ar[r]^\rho  &  GL(V) \\
    g  \ar[u]^{\exp}\ar[r]^{\rho_*} &  gl(V)\ar[u]_{\exp}
    }
    $$
    is commutative, the restriction of the Lie group rep to $SO(3)$ corresponds to the restriction of the Lie algebra rep to $span{\mathbf{J}}$, which is $\pi_A\otimes\pi_B$. So $A+B$ is half odd iff the subrep of $SO(3)$ is projective. As the kernel of $SU(2)\to SO(3)$ is $\pm I$, it is equivalent, when considering rep of $SU(2)$, to $-I$ mapping to $-I|_V$ (by Schur's lemma the center maps to scalars of the identity, and $(-I)^2=I$). Again as the kernel of $SL(2,\mathbb{C}) \to SO^+(1,3)$ is $\pm I$, it is equivalent to the rep of $SO^+(1,3)$ is projective. 

    In general case, it follows by that the direct sum of $\pm I|_{V_i}$ is $I|_{\oplus V_i}$ iff they are all $+I$, and is $-I$ iff they are all $-I$. And we know that an ordinary rep of $SL(2,\mathbb{C})$ can descend to a proj(and not ordinary) rep of $SO^+(1,3)$ iff $-I$ maps to $-I$. 
\end{proof}

\begin{eg}$\quad$

    Scalar representation has $(A,B)=(0,0)$. So it is an ordinary rep of $SO^+(1,3)$. 

    Vector representation has $(A,B)=(\frac{1}{2}, \frac{1}{2})$.
    \footnote{Ironically, writing in standard basis these are not the same. But one can show that it can be related by a similarity transformation. }
    So it is an ordinary rep of $SO^+(1,3)$. 

    Dirac representation has $(\frac{1}{2}, 0) \oplus (0, \frac{1}{2})$. So it is a projective rep of $SO^+(1,3)$. 

    $(\frac{1}{2}, 0) \oplus (0, 1)$ for example, is not even a projective rep of $SO^+(1,3)$. It is just an ordinary rep of $SL(2,\mathbb{C})$. 
\end{eg}

\subsection{Induced representation}

We refer to \cite{etingof2011introduction} and \cite{serre1977linear}.

Induced representation is an important tool for the construction of representations. In the state representation as we shall see, the kind of representations always come from induced representations, as is shown in \cite{}. 

We firstly introduce the case of finite groups and finite dimensional representations, since it is rigorous in mathematical sense. Actually our treatment of state representation is not rigorous, but still similar to this. 

Let $G$ be a group and $H$ is a subgroup. Let $(\rho, W)$ be a f.d. rep of G and $V$ be a stable subspace (subrep) of $H$ and $\pi$ it restricted representation. 
It is easy to see that $gV$ only depends on the representative class of $g\in G/H$. 
Let $g_i$ be a representative class for $G/H$.

\begin{defi}
    We say that $(\rho, W)$ is induced from $(\pi, V)$ if 
    $$
    W = \bigoplus_{i=1}^n g_iV
    $$
\end{defi}
Once this condition is satisfied, we can compute the action of $G$:
$$
g\cdot  g_iv_i =  g_{j(i)} \pi(h_i)v_i
$$
which follows by writing $gg_i$ uniquely in the form $gg_i=g_{j(i)}h_i$. So the existence of the induced representation is unique. 

Actually we can give a constructive definition as follows, which shows the existence. 

\begin{defi}
    Given a rep $(\pi, V)$ of $H$. It induced rep $Ind_H^G \pi$ is defined as $G$ acting on 
    $$
    W \equiv \bigoplus_{i=1}^n g_iV
    $$
    as follows
    $$
    \rho(g)\cdot \sum_{i=1}^n g_iv_i = \sum_{i=1}^n g_{j(i)} \pi(h_i)v_i
    $$
    where for each $g_i$ we find a unique pair of $h_i\in H$ and $j(i)\in\{1,\cdots,n\}$ such that $gg_i=g_{j(i)h_i}$. 
\end{defi}
\begin{proof}
    We need to show that this is indeed a representation. Indeed we only need to show $\rho(g')\rho(g)g_iv = \rho(g'g)g_iv$. Writing $gg_i=g_{j(i)}h_i$ and $g'g_{j(i)} = g_{k(i)}h_i'$, so $g'gg_i=g_{k(i)}h_i' h_i$. So 
    \begin{equation*}
        \rho(g')\rho(g) g_iv = \rho(g') g_{j(i)}\pi(h_i)v = g_{k(i)}\pi(h_i')\pi(h_i)v
    \end{equation*}
    and 
    \begin{equation*}
        \rho(g'g) g_iv =  g_{k(i)}\pi(h_i'h_i)v
    \end{equation*}
    So these are equal. 
\end{proof}
\begin{rmk}\label{cannot induce}
    From the proof we see a serious problem: projective representation cannot be induced in general! This is because when acting on $\sum_{i=1}^n g_iv_i$, the phase $\pi(h_i')\pi(h_i)=\lambda_i\pi(h_i'h_i)$ are not equal, so putting together, though $\pi$ is a projective rep of $H$, $\rho$ is not a projective rep of $G$. 
\end{rmk}

Above is the construction for induced representations. Below is the property for homomorphisms related to induced representations. 

\begin{lemma}(Frobenius reciprocity)\label{Frobenius reciprocity}
\footnote{See \cite{etingof2011introduction}'s exercise.}\\
     We have the following isomorphism
     \begin{equation}
     \begin{aligned}
         Hom_H(\pi, Res^G_H D) &\cong Hom_G(Ind_H^G \pi, D)\\
         F &\mapsto (g_iv \mapsto D(g_i)F(v))\\
         G|_V &\leftarrow G
     \end{aligned}
     \end{equation}
\end{lemma}
\begin{proof}
    They are clearly mutually inverse. And the first mapping is well defined. 
\end{proof}

\clearpage

\section{Quantum Hilbert space}

\subsection{Single particle space}

Denote by $\mathcal{H}$ the Hilbert space of the states of single particles. Below is the requirement of the Lorentz invariance of the quantum theory.

\begin{ax}(Single particle state)
\footnote{The condition on faithfulness can be replaced by other weaker conditions. }

    The Hilbert space $\mathcal{H}$ of single particle states admits a projective, faithful, irreducible unitary representation of the Poincare group 
    $$
    \tilde{U} : \mathbb{R}^{1,3} \rtimes SO^+(1,3) \to \mathcal{PU}(\mathcal{H})
    $$
\end{ax}

Our mission of this subsection is to study such kind of representations and to show that it can be reduced to the study of a so called 'little group'. 
This is the famous result called Wigner's classification. 

By the remark \ref{cannot induce}, we do not follow the path of Weinberg! 
We denote by $(id, \lambda) : \mathbb{R}^{1,3}\rtimes SL(2,\mathbb{C})\to \mathbb{R}^{1,3}\rtimes SO^+(1,3)$ the universal cover of Poincare group. Thanks to Bargmann's theorem \ref{Bargmann}, the mission is to find faithful except $-I$, unitary ordinary irrep of its universal cover:
\begin{equation}
    U: \mathbb{R}^{1,3}\rtimes SL(2,\mathbb{C}) \to \mathcal{U(H)}
\end{equation}

Firstly, as $\mathbb{R}^{1,3}$ is commutative we can choose a formal basis 
\footnote{This is the first non-rigorous step on the physical path of QFT. But as we have a strong favor of delta functions instead of measures, we just follow this path. More discussions can be found, for example, in \cite{straumann2008unitary}. It involves SNAG theorem. }
\begin{equation}
\{\Psi_{p,\sigma} | p\in A\subset \mathbb{R}^4, \sigma\in\text{a finite set $B_p$ dependent of p}\}
\end{equation}
which satisfies, on the Lie algebra level,
\begin{equation}
    P^\mu \Psi_{p, \sigma}=p^\mu \Psi_{p , \sigma}
\end{equation}
which written in Lie group level, reads
\begin{equation}
    U(I,a) \Psi_{p, \sigma} = e^{i p\cdot a} \Psi_{p , \sigma}
\end{equation}

Write $\mathcal{H}_p = span\{ \Psi_{p, \sigma} | \sigma\in B_p \} $ the common eigenspace of $P^\mu$ with eigenvalue $p^\mu$. 
Or equivalently, the common eigenspace of $U(I,a)$ with eigenvalue $e^{-ia\cdot p}$. 
\footnote{Note that 'eigenspace of $\tilde{U}(I,a)$ with eigenvalue $e^{-ia\cdot p}$' is not well defined, as it is a projective rep. Again we see the power of Bargmann's theorem. Weinberg in \cite{weinberg2002quantum} misuses these. }

The non-homogeneous part of the Poincare group is done, and the remaining part is to determine the homogeneous part.

\begin{lemma}$\quad$

    We claim that: for any $p^\mu \in \mathbb{R}^4$ we can write
    \begin{equation}
        p^\mu = \lambda(L^\mu{ }_\nu (p)) k^\nu
    \end{equation}
    where $L^\mu{ }_\nu$ is in $SL(2,\mathbb{C})$ and it later will be called standard Lorentz transformation, and $k^\nu$ is a standard momentum associated to $p^\mu$ which is defined in the table below.
\end{lemma}
\begin{proof}
The choice are discussed in the next subsection. 
\end{proof}
\begin{tabular}{ccccc}% 其中，tabular是表格内容的环境；c表示centering，即文本格式居中；c的个数代表列的个数
\toprule %[2pt]设置线宽     
 & Mass shell  &  Standard momentum & Small group & Name\\ %换行
\midrule %[2pt]  
(a)& $p^2=-M^2<0, p^0>0$ & $(M,0,0,0)$ & $\cancel{S O(3)} SU(2)$ & Massive\\
(b)& $p^2=-M^2<0, p^0<0$ & $(-M,0,0,0)$ & $\cancel{S O(3)} SU(2)$ & \\
(c)& $p^2=0, p^0>0$ & $(1,0,0, 1)$ & $\cancel{I S O(2)}$ & Massless\\
(d)& $p^2=0, p^0<0$ & $(-1, 0,0, 1)$ & $\cancel{I S O(2)}$ & \\
(e)& $p^2=N^2>0$ & $(0,0, 0, N)$ & $\cancel{S O(2,1)}$ & Tachyonic\\
(f)& $p^\mu=0$ & $(0,0, 0, 0)$ & $\cancel{S O^+(1,3)} SL(2,\mathbb{C})$ & Vacuum\\
\bottomrule %[2pt]
\end{tabular}
%%%%%%%%%%%%%%%%%%%%%%%%%%%%%%%%%%%%%%%%%%%%%%%%%%%%%%%%%%%

\begin{defi}(Little group, mass shell)\\
    Let $k$ be a standard momentum as shown above.
    \footnote{In this article the first component denoted the time component, differing from \cite{weinberg2002quantum} who use the last to denote time.} 
    The subgroup
    \begin{equation}
        \mathcal{W}_k \equiv \{A\in SL(2,\mathbb{C}) | \lambda(A)k=k\}
    \end{equation}
    that fix $k$ is called the little group of $k$.
    \footnote{We differ from \cite{weinberg2002quantum} again on the definition of little groups. }
    The orbit of $k$ of the action of $SO^+(1,3)$ is called the mass shell of $k$.
\end{defi}

\begin{prop}$\quad$\\
    For any $p\in A$, 
    \begin{equation}
        U(\Lambda, 0) \mathcal{H}_p
        = \mathcal{H}_{\lambda(\Lambda) p}, \quad \forall \Lambda\in SL(2,\mathbb{C})
    \end{equation}
    Or, equivalently, 
    \begin{equation}
        \tilde{U}(\Lambda, 0) \mathcal{H}_p
        = \mathcal{H}_{\Lambda p}, \quad \forall \Lambda\in SO^+(1,3)
    \end{equation}
\end{prop}
\begin{proof}
On projective rep level,
$$
\tilde{U}(I,a) U(\Lambda, 0) \Psi_{p,\sigma} = \square \tilde{U}(\Lambda, 0) \tilde{U}(I, \Lambda^{-1}a) \Psi_{p,\sigma}
= \square e^{i\Lambda^{-1}a \cdot p} \tilde{U}(\Lambda, 0) \Psi_{p,\sigma}
= \square e^{i a \cdot \Lambda p} \tilde{U}(\Lambda, 0) \Psi_{p,\sigma}
$$
where $\square$ is the phase coming from projectivity. So we can conclude nothing from this. 
Indeed we should compute at the linear rep level
$$
U(I,a) U(\Lambda, 0) \Psi_{p,\sigma} 
= U(\Lambda, 0) U(I, \lambda(\Lambda)^{-1}a) \Psi_{p,\sigma}
= e^{i\lambda(\Lambda)^{-1}a \cdot p} U(\Lambda, 0) \Psi_{p,\sigma}
= e^{i a \cdot \lambda(\Lambda) p} U(\Lambda, 0) \Psi_{p,\sigma}
$$
for $\Lambda\in SL(2,\mathbb{C})$. 
So we have a $\subset$ relation, and the converse is true by applying the equation to $U(\Lambda^{-1},0)$. 
\end{proof}

With these preparations we are ready to prove the famous result:
\begin{thm}(Wigner's classification)\\
Every faithful(except $(-I,0)$), irreducible unitary representation $U$ of the universal cover of Poincare group $\mathbb{R}^{1,3} \rtimes SL(2,\mathbb{C})$, when restricted to the universal cover of Lorentz group $SL(2,\mathbb{C})$, must be an induced representation of a small group $\mathcal{W}(k)$ for a certain standard momentum $k$. 

% More specificly, the representation of $SO^+(1,3)$ on 
% $span\{\Psi_{p,\sigma}|p\in\text{the shell of k}, \sigma\}$ is the induced representation of $\mathcal{W}(k)$ on $span\{\Psi_{k,\sigma}|\sigma\}$. 
\end{thm}

\begin{proof}

Step 1.

By the previous proposition and the diagonal action of $U(I, a)$, any invariant subspace must look like 
$$
\sum_{p\in C} \mathcal{H}_p
$$
where $C$ is a union of certain mass shells. As the projective rep is faithful, C contains at least one non-zero mass shell. So as the projective rep is irreducible, C is a single non-zero mass shell. 
Denote by $k$ its standard momentum. In summary, we have 
$$
\mathcal{H} = \sum_{p\in\text{the mass shell of k}} \mathcal{H}_p
$$ 
for a certain standard momentum $k$.

Step 2.

By the previous proposition again, $\mathcal{H}_k$ is invariant under $\mathcal{W}(k)$. 

Step 3.

$$
\{L(p)|p\in \text{the shell of k}\}
$$ 
constitutes a representative system of $\mathcal{W}(k)$ in $SL(2,\mathbb{C})$. 
This is because
$L(p)^{-1} \Lambda \in \mathcal{W}(k)$ is equivalent to $\lambda(\Lambda) k = \lambda(L(p)) k = p$. So it follows the uniqueness and existence of $L(p)$.

Step 4.

Again, and notice that different eigenspaces of $P^\mu$ must be a direct sum. So we have 
\begin{equation}
    \mathcal{H} = \bigoplus_{p\in\text{the mass shell of k}} U(L(p),0) \mathcal{H}_k
\end{equation}
And it is just the definition for induced representation. 
\end{proof}

\begin{rmk}
    Actually from the proof we see that below is also correct
    \begin{equation}
    \mathcal{H} = \bigoplus_{p\in\text{the mass shell of k}} \tilde{U}(\lambda(L(p)),0) \mathcal{H}_k
    \end{equation}
    So we can roughly say that the projective rep of $\mathbb{R}^{1,3} \rtimes SO^+(1,3)$ is induced from a projective rep of $\lambda(\mathcal{W}(k))$, the little group in Weinberg sense. But this does not imply the uniqueness and existence of such induction.
    Actually this is packed into Bargmann's theorem \ref{Bargmann} in our discussion, which allows us to pass to the universal cover, and where the existence and uniqueness of induction follows directly as in Section 1.5, Frobenius induction. Apparently reducible rep of small group induces reducible ones for large group. It only remains whether irr rep of small group induce irr rep of the large group. This is Mackey's theory. See \cite{straumann2008unitary}.
\end{rmk}

Combining the above results, we can choose a standard $k\in A$ and make the change of notations
\begin{equation}
    \Psi_{p, \sigma} \equiv N(p) U(L(p),0) \Psi_{k, \sigma}
\end{equation}
where $N(p)$ is a numerical normalization factor, to be chosen later. Actually we can show that the indices $\sigma$ lie in a finite space. This is due to the following fact in the massive case. 

\begin{lemma}(Finiteness of the spin index)
\footnote{See \cite{folland2016course}.}\\
    Every irreducible representation of a compact group is finite dimensional. 
\end{lemma}

%%%%%%%%%%%%%%%%%%%%%%%%%%%%%%%%%%%%%%%%%%%%%%%%%%%%%%%%%%%%%%%%%%%%%%%%%%%
The following procedure is just the general procedure for writing down the induced representations. 

\begin{defi}(Representation of little group)\\
We write $D: \mathcal{W} \to \mathcal{U}(\mathcal{H}_k)$ the linear unitary representation of the little group (subgroup of the double cover) restricted from $U$. Then
\begin{equation}
U(W) \Psi_{k, \sigma}
=\sum_{\sigma'} D_{\sigma' \sigma}(W) \Psi_{k, \sigma'}
\quad\quad \forall W\in \mathcal{W}(k)
\end{equation}
\end{defi}

By direct computation we have
\begin{equation}
    \begin{aligned}
    U(\Lambda, a) \Psi_{p, \sigma} 
    & = e^{-ip\cdot a} N(p) U(\Lambda L(p)) \Psi_{k, \sigma} \\
    & = e^{-ip\cdot a} N(p) U(L(\lambda(\Lambda) p)) U(L^{-1}(\lambda(\Lambda) p) \Lambda L(p)) \Psi_{k, \sigma} 
    \end{aligned}
\end{equation}
where we use the abbreviation $U(\Lambda,0) = U(\Lambda)$ from now on.

We now write
\begin{equation}
W(\Lambda, p) \equiv L^{-1}(\lambda(\Lambda) p) \Lambda L(p)
\end{equation}
and then
\\\\
\fbox{%
  \parbox{1\textwidth}{
\begin{equation}\label{single particle state}
\begin{aligned}
U(\Lambda, a) \Psi_{p, \sigma}
&= e^{-ip\cdot a} N(p) \sum_{\sigma'} D_{\sigma' \sigma}(W(\Lambda, p)) U(L(\lambda(\Lambda) p)) \Psi_{k, \sigma'}
\\
&= e^{-ip\cdot a} 
(\frac{N(p)}{N(\lambda(\Lambda) p)}) \sum_{\sigma'} D_{\sigma' \sigma}(W(\Lambda, p)) \Psi_{\lambda(\Lambda) p, \sigma'}
\end{aligned}
\end{equation}
}
}
\\\\

The above formula can be regarded as the main result of this subsection.

We may choose the Lorentz invariant normalization 
\begin{equation}
    \begin{aligned}
        % \left\{
        N(p) &= \sqrt{m/p^0} \quad\text{massive case}\\
        N(p) &= \sqrt{1/p^0} \quad\text{massless case}
        % \right.
        \\
        (\Psi_{p,\sigma}, \Psi_{p', \sigma'}) &= \delta_{\sigma,\sigma'} \delta^3(\mathbf{p'}-\mathbf{p})
    \end{aligned}
\end{equation}

\subsection{Little group--Spin and Helicity of states}

\subsubsection{Massive particles}

In this case the small group is $SU(2)$. As we have seen in the previous section, its irr rep is indexed by a single non-negative half-integer $j$.

\begin{defi}(Spin of state)\\
    The half-integer $j$ is called the spin of the states. 
\end{defi}

In this case the standard transformation $L(p)$ can be chosen as:

\begin{defi}(Standard Lorentz transformation)\\
On Lorentz group level standard transformation can be chosen as
\begin{equation}
    \lambda(L(p)) \equiv
    \begin{pmatrix}
        \gamma & \frac{p_1}{M} & \frac{p_2}{M} & \frac{p_3}{M} \\
         \frac{p_1}{M} & 1 + (\gamma-1)\frac{p_1 p_1}{|\vec{p}|^2} & (\gamma-1)\frac{p_1 p_2}{|\vec{p}|^2} & (\gamma-1)\frac{p_1 p_3}{|\vec{p}|^2} \\
         \frac{p_2}{M} & (\gamma-1)\frac{p_2 p_1}{|\vec{p}|^2} & 1 + (\gamma-1)\frac{p_2 p_2}{|\vec{p}|^2} & (\gamma-1)\frac{p_2 p_3}{|\vec{p}|^2} \\
         \frac{p_3}{M} & (\gamma-1)\frac{p_3 p_1}{|\vec{p}|^2} & (\gamma-1)\frac{p_3 p_2}{|\vec{p}|^2} & 1 + (\gamma-1)\frac{p_3 p_3}{|\vec{p}|^2}
    \end{pmatrix}
    ,\quad
    \gamma\equiv p^0/M
\end{equation}
It is easy to verify that it is Lorentz and it takes $(M,0,0,0)$ to $p^\mu$.  
Define $\cosh\theta \equiv \gamma = p^0/M$ and $\widehat{p}$ is the unit-normalized 3-vector of $\vec{p}$. We see
\begin{equation*}
    L(\theta) = 
        \begin{pmatrix}
            \cosh\theta  &  \widehat{p_1}\sinh\theta & \widehat{p_2}\sinh\theta & \widehat{p_3}\sinh\theta
            \\
            \widehat{p_1}\sinh\theta  &  1+ (\cosh\theta-1)\widehat{p_1}\widehat{p_1}
            &  (\cosh\theta-1)\widehat{p_1}\widehat{p_2} & (\cosh\theta-1)\widehat{p_1}\widehat{p_3}
            \\
            \widehat{p_2}\sinh\theta  &  (\cosh\theta-1)\widehat{p_2}\widehat{p_1}
            &  1+(\cosh\theta-1)\widehat{p_2}\widehat{p_2} & (\cosh\theta-1)\widehat{p_2}\widehat{p_3}
            \\
            \widehat{p_3}\sinh\theta  &  (\cosh\theta-1)\widehat{p_3}\widehat{p_1}
            &  (\cosh\theta-1)\widehat{p_3}\widehat{p_2} & 1+(\cosh\theta-1)\widehat{p_3}\widehat{p_3}
        \end{pmatrix}
    \end{equation*}
    Then we can verify $L(\theta')L(\theta) = L(\theta'+\theta)$. So $\theta\mapsto L(\theta)$ is a one-parameter subgroup for the Lie group $SO^+(1,3)$. And that 
    $\frac{d}{d\theta}L(\theta)|_{\theta=0}=\widehat{\mathbf{p}} \cdot \mathbf{K} \theta$ is obvious. So by the commutativity of exponential map with Lie group-Lie algebra morphism, its universal covering level can be chosen as
    \begin{equation}
        L(p) = \exp(-i \widehat{p}\cdot \vec{K}\theta)
    \end{equation}
    where we identify $so(1,3)$ as the Lie algebra of $SL(2,\mathbb{C})$ and the exponential is taken to be that of $SL(2,\mathbb{C})$. 
\end{defi}

\subsubsection{Massless particles}

The little group on the Poincare group level is discussed in \cite{weinberg2002quantum}, which is $\lambda(\mathcal{W})=ISO(2)$, in which the general group elements are
\begin{equation}
    \tilde{W}(\theta, \alpha, \beta) = \tilde{S}(\alpha, \beta) \tilde{R}(\theta)
\end{equation}
where
\begin{equation}
    \tilde{S}(\alpha, \beta) \equiv
    \begin{pmatrix}
        1 + \zeta & \alpha & \beta & -\zeta\\
         \alpha & 1 & 0 & -\alpha\\
         \beta & 0 & 1 & -\beta \\
         \zeta & \alpha & \beta & 1 -\zeta \\
    \end{pmatrix}, \quad
    \zeta \equiv (\alpha^2+\beta^2)/2
\end{equation}

\begin{equation}
    \tilde{R}(\theta) \equiv
    \begin{pmatrix}
        1 & 0 & 0 & 0\\
        0 & \cos\theta & \sin\theta & 0 \\
        0 & -\sin\theta & \cos\theta & 0 \\
        0 & 0 & 0 & 1
    \end{pmatrix}
\end{equation}

Actually this can be easily seen at the level of the universal covering. Recall that the universal mapping is defined by
\begin{equation*}
    P(A): h\to h, X\mapsto A^\dagger XA
\end{equation*}
where we identify $h$ with $\mathbb{R}^{1,3}$ by 
\begin{equation*}
    X = 
    \begin{pmatrix}
        t+z & x+iy\\
        x-iy & t-z
    \end{pmatrix}
    \mapsto (t,x,y,z)
\end{equation*}

Solving
$$
\begin{pmatrix}
    x^* & z^*\\
    y^* & w^*
\end{pmatrix}
\begin{pmatrix}
    2 & 0\\
    0 & 0
\end{pmatrix}
\begin{pmatrix}
    x & y\\
    z & w
\end{pmatrix}
=
\begin{pmatrix}
    2 & 0\\
    0 & 0
\end{pmatrix}
$$
we get $|x|=1,y=0$. So 
\begin{equation}
    \mathcal{W} = \{ 
    W(\theta, \alpha, \beta) \equiv 
    \begin{pmatrix} 
    e^{-i\theta/2} & 0\\
    e^{-i\theta/2}(\alpha-i\beta) & e^{i\theta/2}
    \end{pmatrix}
    \}
\end{equation}
And one can verify 
$$
\lambda(
W(\theta, \alpha, \beta)
)
=
\tilde{W}(\theta, \alpha, \beta)
$$
by checking its action on $h$. defining $S(\alpha,\beta)=W(0,\alpha,\beta), R(\theta)=W(\theta,0,0)$, we have
\begin{equation}
    \begin{aligned}
    S(\alpha',\beta') S(\alpha,\beta) &= S(\alpha' + \alpha,\beta' + \beta)\\
    R(\theta') R(\theta) &= R(\theta' + \theta)\\
    R(\theta) S(\alpha,\beta) R^{-1}(\theta) 
    &= S(\alpha\cos\theta +\beta\sin\theta, -\alpha\sin\theta +\beta\cos\theta)
    \end{aligned}
\end{equation}

The Lie algebra is spanned by
\footnote{Weinberg in \cite{weinberg2002quantum} has a typo.}
\begin{equation}
    \begin{aligned}
        J_2-K_1\\
        -J_1-K_2\\
        J_3
    \end{aligned}
\end{equation}

For some reason, we only consider
\footnote{See \cite{weinberg2002quantum}}
the one-dimensional representation of $\mathcal{W}$ defined by
\begin{equation}
    D^{\sigma}(W(\theta, \alpha, \beta)) = \exp(i\theta \sigma)
\end{equation}
In order to make it well defined, it is equivalent to $\sigma$ being a half integer.

\begin{defi}(Helicity of state)\\
    The half integer $\sigma$ is called the helicity of states.  
\end{defi}

\begin{rmk}
    The assumption that we only consider the representation of $\mathcal{W}$ that is trivial on $S(\alpha,\beta)$'s is intriguing. Maybe I will discuss it in a later article. 
\end{rmk}

\begin{defi}(Standard Lorentz transformation)\\
    For a momentum $p^\mu$ on the mass shell, on Lorentz group level we can choose
    \begin{equation}
        \lambda(L(p)) = R(\widehat{p}) B(|\vec{p}|)
    \end{equation}
    where $R(\widehat{p})$ is a rotation from z-axis to $\widehat{p}$ and $B(|\vec{p}|)$ is the standard boost along z-axis:
    \begin{equation}
        B(u) \equiv
        \begin{pmatrix}
            1 & 0 & 0 & 0\\
            0 & 1 & 0 & 0\\
            0 & 0 & (u^2+1)/2u & (u^2-1)/2u\\
            0 & 0 & (u^2-1)/2u & (u^2+1)/2u
        \end{pmatrix}
    \end{equation}
     Write
    \begin{equation}
        p=(|p|, |p|\sin\theta\cos\phi, |p|\sin\theta\sin\phi, |p|\cos\theta)
    \end{equation}
    and on universal covering level we can choose
    \begin{equation}
        L(p) = \exp(-i\phi J_3) \exp(-i\theta J_2) \exp(-i \ln(|\vec{p}|) K_3)
    \end{equation}
    where we identify $so(1,3)$ as the Lie algebra of $SL(2,\mathbb{C})$ and the exponential is taken to be that of $SL(2,\mathbb{C})$. 
\end{defi}

\subsection{Fock space, creation and annilation operators}

\begin{ax}(Fock space)\\
    The whole Hilbert space is the Fock space defined below
    \begin{equation}
        \mathcal{F} = \mathcal{F}_\nu(\mathcal{H}) = \overline{\bigoplus_{n=0}^\infty S_\nu(\mathcal{H}^{\otimes n}) }
    \end{equation}
    Here $\nu=\pm$ and $S_\nu$ is the symmetric or anti-symmetric quotient of a space, and overline denotes Hilbert space completion. 
    
    It is equipped with a representation of $\mathbb{R}^{1,3}\rtimes SL(2,\mathbb{C})$, constructed by the representation on the single particle space via commutative/anti-commutative square and direct sum. Notice that it do not carry a projective representation of the Poincare group\footnote{It is misstated in \cite{weinberg2002quantum}. It cost me a lot of time to realize it.}, even if the single particle space does. This comes from the fact that projective representations cannot make direct sums in general. 
\end{ax}

It is interesting to see that the division into boson and fermion are at the beginning of the theory, entered as an axiom. 

From now on we denote by $D^{(j)}$, instead of $D$, the (finite dimensional) linear representation of the little group $\mathcal{W} = \mathcal{W}(k)$ on  
$\mathcal{H}_k$. The letter $D$ is left to another representation of the Lorentz group, as we will meet in section 3. Also, from now on we denote by $U_0$, instead of $U$, the projective representation in Axiom 1. The reason for saving the letter $U$ is when considering $S$ matrix, we meet two representations satisfying Axiom 1, corresponding to the whole Hamiltonian $H$ and the free Hamiltonian $H_0$ respectively. 

We just list all the results needed in the remaining part of the article.

\begin{defi}$\quad$\\
    Define creation operators as
    \begin{equation}
    a^{\dagger}(q) \Phi_{q_1 q_2 \cdots q_n} \equiv \Phi_{q q_1 q_2 \cdots q_N}
    \end{equation}
    Its adjoint is then called $a(q)$, the annilation operators. 
\end{defi}
\begin{prop}$\quad$\\
    When the particles $q q_1 \cdots q_N$ are either all bosons or all fermions, we have
    \begin{equation}
    a(q) \Phi_{q_1 q_2 \cdots q_N}=\sum_{r=1}^N(\pm)^{r+1} \delta(q-q_r) \Phi_{q_1 \cdots q_{r-1} q_{r+1} \cdots q_N},
    \end{equation}
\end{prop}
\begin{prop}(Commutation or anti-commutation relation of creation and annilation)\label{commutation of CA}
    \begin{equation}
    \begin{aligned}
        a(q') a^{\dagger}(q) \mp a^{\dagger}(q) a(q') &=\delta(q'-q)\\
        a^{\dagger}(q') a^{\dagger}(q) \mp a^{\dagger}(q) a^{\dagger}(q') &=0\\
        a(q') a(q) \mp a(q) a(q') &=0
    \end{aligned}
    \end{equation}
\end{prop}

\begin{prop}(Lorentz transformation of creation and annilation)\label{Lorentz of CA}\\
    In order that the state (4.2.2) should transform properly, it is necessary and sufficient that the creation operator have the transformation rule
    \begin{equation}
        U_0(\Lambda, \alpha) a^{\dagger}(\mathbf{p} \sigma n) U_0^{-1}(\Lambda, \alpha)
        = e^{-i(\lambda(\Lambda) p) \cdot \alpha} \sqrt{(\lambda(\Lambda) p)^0 / p^0}
        \times \sum_{\bar{\sigma}} 
        D_{\bar{\sigma} \sigma}^{(j)}(W(\Lambda, p)) a^{\dagger}(\mathbf{p}_{\Lambda} \bar{\sigma} n)
    \end{equation}
\end{prop}

\subsection{Discussions}

\begin{itemize}

\item

    Unlike what Weinberg have done to single particle state, where we can still write down a projective version, when writing down the Lorentz transformation of creation and annilation operators we do not have a projective version at all. Because we cannot have a projective rep of Poincare group on Fock space, either from constructing direct sum representation or from descending linear rep of $\mathbb{R}^{1,3}\rtimes SL(2,\mathbb{C})$ on Fock space. In summary:

    
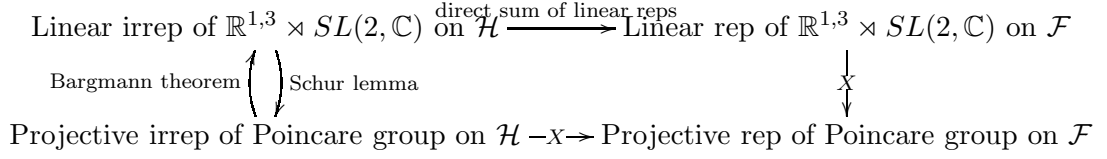
\begin{equation*}
        \xymatrix{
        \text{Linear irrep of $\mathbb{R}^{1,3}\rtimes SL(2,\mathbb{C})$ on $\mathcal{H}$} 
        \ar@/^/[d]^{\text{Schur lemma}}\ar[r]^{\text{direct sum of linear reps}}
        & \text{Linear rep of $\mathbb{R}^{1,3}\rtimes SL(2,\mathbb{C})$ on $\mathcal{F}$}\ar[d]|X\\
        \text{Projective irrep of Poincare group on $\mathcal{H}$}\ar@/^/[u]^{\text{Bargmann theorem}} \ar[r]|X
        & \text{Projective rep of Poincare group on $\mathcal{F}$}
        }
    \end{equation*}

    Consider $-I$ in universal cover for example. If it represents to $-I$ on single particle space, then it represents to $+I$ on 2-particles space. So it does not direct sum into a phase. So we cannot descend it into a projective rep of Poincare group.

\item

    I can explain in detail why we cannot construct direct sum for projective reps. For two reps $\rho_1:G\to GL(V_1),\rho_2:G\to GL(V_2)$, define its direct sum
    \begin{equation}
        \rho(g) = \rho_1(g) \oplus \rho_2(g)
    \end{equation}
    When they are linear reps,
    \begin{equation}
        \rho(gh) = \rho_1(gh) \oplus \rho_2(gh) = \rho_1(g)\rho_1(h) \oplus \rho_2(g)\rho_2(h)
        =\rho(g)\rho(h)
    \end{equation}
    But when they are just projective reps, 
    \begin{equation}
    \begin{aligned}
        \rho(gh)& = \rho_1(gh) \oplus \rho_2(gh) = \omega\rho_1(g)\rho_1(h) \oplus
        \omega' \rho_2(g)\rho_2(h)
        \\
        \rho(g)\rho(h)&=\rho_1(g) \oplus \rho_2(g) \cdot \rho_1(h) \oplus \rho_2(h)
        =\rho_1(g)\rho_1(h) \oplus \rho_2(g)\rho_2(h)
    \end{aligned}
    \end{equation}
    Two phases on two proportions do not always equal, so the LHS and RHS are not equal up to a scalar.

\item
    What is the physical meaning of the representation of Poincare group on state space?

    Let $M$ be the spacetime manifold. A reference frame is a coordinate on the manifold: $\varphi:M\to \mathbb{R}^{1,3}$\footnote{I haven't seen this physical picture involving commutative diagrams. I came up with it alone, motivated by general relativity which uses the language of smooth manifolds. }. The coordinate transformation between two reference frames is, by definition, $\varphi_\beta \circ \varphi_\alpha^{-1}:\mathbb{R}^{1,3}\to \mathbb{R}^{1,3}$.
    $$
    
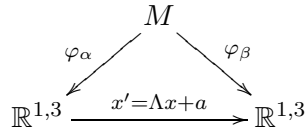
\xymatrix{    
    &M \ar[ld]_{\varphi_\alpha} \ar[rd]^{\varphi_\beta}
    \\
    \mathbb{R}^{1,3} \ar[rr]^{x'=\Lambda x+a}
    & &\mathbb{R}^{1,3}
    }
    $$
    In Schrodinger's picture, there is an abstract state space $\tilde{\mathcal{H}}$, while what we have used is the concrete state space $\mathcal{H}$. Given a coordinate on $M$, $\varphi:M\to \mathbb{R}^{1,3}$, there is a trivialization $U(\varphi): \tilde{\mathcal{H}} \to \mathcal{H}$. Its physical meaning is that for an abstract state in the abstract state space $\tilde{\mathcal{H}}$, what the person in that reference frame measured as a concrete state in $\mathcal{H}$.
    So the physical meaning of the state rep $U$ is: if the coordinate transformation between two reference frames is $\varphi_\beta \circ \varphi_\alpha^{-1}(x)=\Lambda x+a$, for an abstract state $\tilde{\psi}$, the two measured concrete states $\psi$ and $\psi'$ is related by $\psi'=U(\Lambda,a)\psi$. 
    $$
    \xymatrix{         
    &\tilde{\psi} \in \tilde{\mathcal{H}} \ar[ld]_{U(\varphi_\alpha)} \ar[rd]^{U(\varphi_\beta)}
    \\
    \psi \in \mathcal{H} \ar[rr]^{U(\Lambda,a)}
    & &\mathcal{H} \ni \psi'=U(\Lambda,a)\psi
    }
    $$
    Schrodinger equation in QM is a special case of the above point of view. Imagine a ‘present man’ and a ‘future man’. ‘Present man’ is in the reference frame $\varphi_\alpha:M\to \mathbb{R}^{1,3}$, ‘future man’ is in $\varphi_\beta:M\to \mathbb{R}^{1,3}$. For a same spacetime point, future man tend to think of its time-coordinate smaller than what the present man thinks (e.g. what the present man thinks of as $t=0$ might be seen as the past for future man, i.e. $t<0$). So the coordinate transformation is $\varphi_\beta\circ\varphi_\alpha^{-1}(t,\mathbf{x}) = (t-\Delta t,\mathbf{x})$. Yeah, it is subtraction instead of addition. So the two measured concrete states of a same abstract state appear to be related by $\psi'=U(I,-\Delta t)\psi = e^{-iH\Delta t}\psi$, where the "'" stands for future man. This is just the well-known Schrodinger equation!
    In the QM's view, Schrodinger's equation is understood as the time-evolution of a state, but now it is understood as the difference between measured results between the present man and the future man.

\item
    What remains not rigorous?
    
    1. Clearly the $\delta$ function involved in the normalization tells us it must be not rigorous. 
    
    2. The first not rigorous step is choosing common eigenvector of $U(I,a)$'s. SNAG's theorem ensures this on the level of measures. 

    3. We are considering the homomorphism of Poincare group to $\mathcal{U(H)}$ that is continuous on the strong operator topology level. But we treated it like we did in finite dimensional case. For example, induced representation should be discussed differently. 

    But we don't discuss these in this article. 
\end{itemize}

\clearpage
\section{Free quantum fields}

Before moving on to the serious content, we briefly talk about the motivation for the construction of quantum fields. Because it is not the main topic of this article, we do not discuss the precise meaning in detail, and refer to \cite{weinberg2002quantum} for comprehensive discussions. 

First we have one Fock space with two state representations $U_0$ and $U$ as in Axiom 1. Also we have three classes of states $\Phi$, $\Psi^-$ and $\Psi^+$, as well as three classes of creation and annilation operators $a$, $a_{in}$ and $a_{out}$. 

The quantity that directly relates to experiments is the S operator, defined by
\begin{equation*}
    (\Phi_\beta, S \Phi_\alpha) \equiv (\Psi^-_\beta, \Psi^+_\alpha)
\end{equation*}
The core of QFT is to make it Lorentz invariant, i.e.,
\begin{equation*}
    U_0(\Lambda, a)^{-1} S U_0(\Lambda, a) = S
\end{equation*}
The famous Dyson series gives a formula for S operator
\begin{equation*}
    S = \mathcal{T} \exp( -i \int_{-\infty}^{+\infty} dt V(t) )
\end{equation*}
In order for this quantity to be Lorentz invariant, we hope to write
\begin{equation*}
    V(t) = \int d^3x \mathscr{H}(\mathbf{x} , t)
\end{equation*}
and hope that
\begin{itemize}
    \item 
    \begin{equation}
        U_0(\Lambda, a)^{-1} \mathscr{H}(x) U_0(\Lambda, a) = \mathscr{H}(\Lambda x +a)
    \end{equation}
    \item 
    \begin{equation}
        [\mathscr{H}(x), \mathscr{H}(x')] = 0, \quad\quad\text{when $(x-x')^2> 0$}
    \end{equation}
\end{itemize}

Above motivates the philosophy for constructing free quantum fields:
\begin{itemize}
    \item 
    In order for the first to satisfy, we need quantum fields--creation and annilation fields $\psi^\mp(x)$, which satisfy a so called 'Lorentz transformation of fields'. We use polynomials of these to construct $\mathscr{H}(x)$. 
    \footnote{Actually we can also reconstruct $H_0$, see \cite{weinberg2002quantum}. But it is strange that we reconstruct something we already have. }
    \item 
    In order for the second to satisfy as well as the $\mathscr{H}(x)$ being Hermitian, we need to combine these two fields into a single one $\psi_l(x) = \kappa_l \psi^+_l(x) + \lambda_l \psi^-_l(x)$ and hope that
    \begin{equation}
        [\psi_l(x), \psi_{l'}(x')]_\mp = 0, \quad\quad\text{when $(x-x')^2> 0$}
    \end{equation}
    \item 
    For the motivation of the existence of a conservative charge $Q$, we find 
    \begin{equation*}
        \begin{aligned}
        [Q, a(\mathbf{p}, \sigma, n)] &= -q(n) a(\mathbf{p}, \sigma, n)\\
        [Q, a^\dagger(\mathbf{p}, \sigma, n)] &= +q(n) a^\dagger(\mathbf{p}, \sigma, n)
        \end{aligned}
    \end{equation*}
    not satisfactory for our purpose, as the polynomial constructed from $\psi$ do not have a simple commutation relation with $Q$. So we replace $a^\dagger$ with $a^{c\dagger}$, which creates particles with charge opposite to that of $a^\dagger$, into the single quantum field. 
\end{itemize}

Above is the soul or outline of this section. 

\begin{rmk}
    The bracket for $\mathscr{H}(x)$ is always commutation bracket, no matter whether the particle is boson or fermion. But the bracket for $\psi_l(x)$ depends on whether it is boson or fermion. The later is because the corresponding is computable using commutation/anti-commutation relations for creation and annilation operators, and once we make sure to always put an even number of fermion fields into the Hamiltonian, this reduces to commutation bracket. 
\end{rmk}

We fix a standard momentum $k$ of type (a) or (c).
Again let $D^{(j)}$ the (finite dimensional) projective representation of the little group $\mathcal{W} = \mathcal{W}(k)$ on  
$\mathcal{H}_k$. We rewrite the main result of the previous section \eqref{single particle state} as
\begin{equation}\label{single particle state 2}
    \begin{aligned}
    U_0(\Lambda, a) \Psi_{p, \sigma}
    &= e^{-ip\cdot a} 
    (\frac{N(p)}{N(\lambda(\Lambda) p)}) \sum_{\sigma'} D_{\sigma' \sigma}(W(\Lambda, p)) \Psi_{\lambda(\Lambda) p, \sigma'}
    \end{aligned}
\end{equation}

\clearpage

\subsection{General settings}

We introduce annihilation fields $\psi_{l}^{+}(x)$ and creation fields $\psi_l^{-}(x)$:
\begin{equation}
    \begin{aligned}
    & \psi_{l}^{+}(x)=\sum_{\sigma n} \int d^3 p u_{l}(x ; \mathbf{p}, \sigma, n) a(\mathbf{p}, \sigma, n), \\
    & \psi_{l}^{-}(x)=\sum_{\sigma n} \int d^3 p v_{l}(x ; \mathbf{p}, \sigma, n) a^{\dagger}(\mathbf{p}, \sigma, n)
    \end{aligned}
\end{equation}
We hope it to satisfy

\begin{ax}(Lorentz transformation of field)
\footnote{Many physical literature claim them to be a linear rep of the Lorentz group. It took me a lot of time to realize that to be wrong. }
\footnote{Our convention is different from that of \cite{weinberg2002quantum}.}

    We have a finite dimensional representation $D$ of the universal cover $SL(2,\mathbb{C})$ of the Lorentz group such that the 'Lorentz transformation of the field' are satisfied:
    \begin{equation}
    \begin{aligned}
        U_0(\Lambda, a)^{-1} \psi_{l}^{+}(\lambda(\Lambda) x+a) U_0(\Lambda, a)
        = D_{l \bar{l}}(\Lambda) \psi_{\bar{l}}^{+}(x)\\
        U_0(\Lambda, a)^{-1} \psi_{l}^{-}(\lambda(\Lambda) x+a) U_0(\Lambda, a)
        = D_{l \bar{l}}(\Lambda) \psi_{\bar{l}}^{-}(x)\\
        \forall \Lambda\in SL(2,\mathbb{C})
    \end{aligned}
    \end{equation}
    Here the indices $l$'s are the indices for the representation $D$. And we have suppresed the summation for $l'$'s. 
\end{ax}

Note that this is the second (projective) representation of the Lorentz group $SO^+(1,3)$. And the space of this rep has nothing to do with the state space $\mathcal{H}$! 
I call the first rep 'the state representation' and the second rep 'the field representation'.

\begin{rmk}
    We see from the definition that the $U_0$ and $D$ qualify the definition for a representation, instead of an opposite of a representation. This is because
    \begin{equation}
    \begin{aligned}
        &U_0(\Lambda_1, a_1)^{-1} U_0(\Lambda_2, a_2)^{-1} 
        \psi_{l}^{+}((\Lambda_2,a_2)(\Lambda_1,a_1)x) 
        U_0(\Lambda_2, a_2) U_0(\Lambda_1, a_1)
        \\
        &=
        U_0(\Lambda_1, a_1)^{-1} 
        D_{l \bar{l}}(\Lambda_2) \psi_{\bar{l}}^{+}((\Lambda_1,a_1)x)
        U_0(\Lambda_1, a_1)
        \\
        &= D_{l \bar{l}}(\Lambda_2) D_{\bar{l} \bar\bar{l}}(\Lambda_1) 
        \psi_{\bar\bar{l}}^{-}(x)
    \end{aligned}
    \end{equation}
\end{rmk}

\begin{rmk}
    Weinberg's book misuses the language of projective or even linear representation of the Lorentz group in the above axiom.
    But as what I've discussed in the last section, there is no such state rep on the Fock space. 

    Further, the above misunderstanding immediately misleads to a deadly conclusion: transforming it twice, the two composition of $U_0$'s is related by an inverse, so the phases offsets. So we wrongly concludes that $D$ must be a linear rep of the Lorentz group. But Section 1.4 showed that Dirac rep is a projective one. So even the admirable Dirac rep is not in our consideration, which is absurd.

    So the radical reason for the mistake is that there is, generally speaking, no such thing as a projective rep of Poincare group on Fock space. 
\end{rmk}

In order to be in accordance with the convention when writing Lorentz transformation of creation and annilation operators, we rewrite Axiom 3 as:
\begin{equation}
\begin{aligned}
& U_0(\Lambda, a) \psi_{l}^{+}(x) U_0^{-1}(\Lambda, a)=\sum_{\bar{l}} D_{l \bar{l}}(\Lambda^{-1}) \psi_{\bar{l}}^{+}(\lambda(\Lambda) x+a)\\
& U_0(\Lambda, a) \psi_{l}^{-}(x) U_0^{-1}(\Lambda, a)=\sum_{\bar{l}} D_{l \bar{l}}(\Lambda^{-1}) \psi_{\bar{l}}^{-}(\lambda(\Lambda) x+a)
\end{aligned}
\end{equation}

The remaining part of this subsection is to determine the condition for $u_l,v_l$'s in order for Axiom 3 to satisfy.

%%%%%%%%%%%%%%%%%%%%%%%%%%%%%%%%%%%%%%%%%%%%%%%%%%%%%%%%%%%%%%%

Recall the Lorentz transformation of creation and annilation operators \ref{Lorentz of CA},
\begin{equation}
        U_0(\Lambda, \alpha) a^{\dagger}(\mathbf{p} \sigma n) U_0^{-1}(\Lambda, \alpha)
        = e^{-i(\Lambda p) \cdot \alpha} \sqrt{(\Lambda p)^0 / p^0}
        \times \sum_{\bar{\sigma}} D_{\bar{\sigma} \sigma}^{(j)}(W(\Lambda, p)) a^{\dagger}(\mathbf{p}_{\Lambda} \bar{\sigma} n)
\end{equation}

By unitarity of the rep, 
\begin{equation}
        U_0(\Lambda, \alpha) a^{\dagger}(\mathbf{p} \sigma n) U_0^{-1}(\Lambda, \alpha)
        = e^{-i(\Lambda p) \cdot \alpha} \sqrt{(\Lambda p)^0 / p^0}
        \times \sum_{\bar{\sigma}} D_{\sigma \bar{\sigma}}^{(j)*}(W^{-1}(\Lambda, p)) a^{\dagger}(\mathbf{p}_{\Lambda} \bar{\sigma} n)
\end{equation}
and taking a dagger,
\begin{equation}
        U_0(\Lambda, \alpha) a(\mathbf{p} \sigma n) U_0^{-1}(\Lambda, \alpha)
        = e^{i(\Lambda p) \cdot \alpha} \sqrt{(\Lambda p)^0 / p^0}
        \times \sum_{\bar{\sigma}} D_{\sigma \bar{\sigma}}^{(j)}(W^{-1}(\Lambda, p)) a(\mathbf{p}_{\Lambda} \bar{\sigma} n)
\end{equation}
\\\\

%%%%%%%%%%%%%%%%%%%%%%%%%%%%%%%%%%%%%%%%%%%%%%%%%%%%%%%%%%%%%%%
\centerline{Step 0}

Taking $\Lambda=I$ in Axiom 3, we have
\begin{equation}
    \psi_l^\pm (x+a) = e^{\pm ip\cdot \alpha} \psi_l^\pm(x)
\end{equation}
So they must take the form
\begin{equation}\label{step2}
\begin{aligned}
& u_{l}(x ; \mathbf{p}, \sigma, n)=(2 \pi)^{-3 / 2} e^{i p \cdot x} u_{l}(\mathbf{p}, \sigma, n) \\
& v_{l}(x ; \mathbf{p}, \sigma, n)=(2 \pi)^{-3 / 2} e^{-i p \cdot x} v_{l}(\mathbf{p}, \sigma, n)
\end{aligned}
\end{equation}
so the fields are Fourier transforms:
\\\\
\fbox{%
  \parbox{1\textwidth}{

\begin{equation}\label{A}
\begin{aligned}
    \psi_l^{+}(x)=\sum_{\sigma, n}(2 \pi)^{-3 / 2} \int d^3 p u_l(\mathbf{p}, \sigma, n) e^{i p \cdot x} a(\mathbf{p}, \sigma, n)
    \\
    \psi_{l}^{-}(x)=\sum_{\sigma, n}(2 \pi)^{-3 / 2} \int d^3 p v_{l}(\mathbf{p}, \sigma, n) e^{-i p x} a^{\dagger}(\mathbf{p}, \sigma, n)
\end{aligned}
\end{equation}
}
}
$\quad$
\\\\

%%%%%%%%%%%%%%%%%%%%%%%%%%%%%%%%%%%%%%%%%%%%%%%%%%%%%%%%%%%%%%%
\centerline{Step 1}
\begin{lemma}(Volume form on the mass shell)\footnote{This understanding avoiding delta functions is found by myself just after studying smooth manifolds.}

    The Lorentz invariant volume form on the mass shell is $md^3p/p^0$. 
\end{lemma}
\begin{proof}
    I present here a mathematical proof without dealing with delta function tricks. 
    The Lorentz invariant volume form on Minkowski manifold $\mathbb{R}^{1,3}$ is $\omega_g = dt\wedge dx\wedge dy\wedge dz$. So by \cite{} P390, the volume form on the mass shell MS is $\omega_{\tilde{g}} = \iota^*(N\lrcorner \omega_g)$, where $N$ is a smooth unit normal vector field along MS determining its orientation. As the defining function for the mass shell is $f(t,x,y,z)=-t^2+x^2+y^2+z^2$, $df/2=-tdt+xdx+ydy+zdz$, and thus $gradf = t\partial_t+x\partial_x+y\partial_y+z\partial_z$. $gradf$ is normal to the regular level set of $f$, i.e. MS. So we can choose $N=gradf/m = (1/m)(t\partial_t+x\partial_x+y\partial_y+z\partial_z)$. So 
    \begin{align*}
    \omega_{\tilde{g}} = \iota^*(N\lrcorner \omega_g) = \iota^*(tdx\wedge dy\wedge dz - xdt\wedge dy\wedge dz + ydt\wedge dx\wedge dz - zdt\wedge dx\wedge dy)/m \\
    = (1/m)(p^0 dx\wedge dy\wedge dz- \frac{x^2+y^2+z^2}{p^0}dx\wedge dy\wedge dz) 
    = \frac{m}{p^0}dx\wedge dy\wedge dz
    \end{align*}
\end{proof}

By Lorentz invariance of $d^3p/p^0$, substituting $d^3p$ with $d^3(\Lambda p)p^0/(\Lambda p)^0$, we find
\begin{equation}
    \begin{aligned}
    & U_0(\Lambda, b) \psi_{l}^{+}(x) U_0^{-1}(\Lambda, b)
    =(2 \pi)^{-3 / 2}
    \sum_{\sigma \textcolor{blue}{\bar{\sigma}} n} \int d^3(\Lambda p) \\
    &\textcolor{blue}
    {u_{l}(\mathbf{p}, \sigma, n)
     \times e^{i p \cdot x}
    \exp (i(\Lambda p) \cdot b) D_{\sigma \bar{\sigma}}^{(j_n)}(W^{-1}(\Lambda, p)) \sqrt{p^0 /(\Lambda p)^0}}
    a(\mathbf{p}_\Lambda, \bar{\sigma}, n)
    \end{aligned}
\end{equation}
and
\begin{equation}
    \begin{aligned}
    & U_0(\Lambda, b) \psi_{l}^{-}(x) U_0^{-1}(\Lambda, b)
    =(2 \pi)^{-3 / 2}
    \sum_{\sigma \textcolor{green}{\bar{\sigma}} n} \int d^3(\Lambda p) \\
    &\textcolor{green}
    {v_{l}(\mathbf{p}, \sigma, n)
     \times e^{-i p \cdot x}
    \exp (-i(\Lambda p) \cdot b) D_{\sigma \bar{\sigma}}^{(j_n) *}(W^{-1}(\Lambda, p)) \sqrt{p^0 /(\Lambda p)^0} }
    a^{\dagger}(\mathbf{p}_{\Lambda}, \bar{\sigma}, n) .
    \end{aligned}
\end{equation}
\begin{prop}$\quad$

    In order for Axiom 3 to satisfy, it is necessary and sufficient that
    \begin{equation}\label{DysonSphere}
        \begin{aligned}
        & \sum_{\bar{l}} D_{l \bar{l}}(\Lambda^{-1}) u_{\bar{l}}(\mathbf{p}_{\Lambda}, \sigma, n)
        =
        \textcolor{blue}{\sqrt{p^0 /(\Lambda p)^0}
        \times 
        u_{l}(\mathbf{p}, \sigma, n)
        \sum_{\bar{\sigma}} D_{\sigma \bar{\sigma}}^{(j_n)}(W^{-1}(\Lambda, p))}
        \\
        & \sum_{\bar{l}} D_{l \bar{l}}(\Lambda^{-1}) v_{\bar{l}}(\mathbf{p}_{\Lambda}, \sigma, n)
        =
        \textcolor{green}{\sqrt{p^0 /(\Lambda p)^0}
        \times 
        v_{l}(\mathbf{p}, \sigma, n)
        \sum_{\bar{\sigma}} D_{\sigma \bar{\sigma}}^{(j_n) *}(W^{-1}(\Lambda, p))}
        \end{aligned}
    \end{equation}
    Or, equivalently,
    \begin{equation}\label{A'}
        \begin{aligned}
        \sqrt{\frac{(\Lambda p)^0}{p^0}}
        \sum_{\bar{\sigma}} 
        u_{\bar{l}} (\mathbf{p}_{\Lambda}, \bar{\sigma}, n) D_{\bar{\sigma} \sigma}^{(j_n)}(W(\Lambda, p))
        = \sum_l D_{\bar{l} l}(\Lambda) u_l(\mathbf{p}, \sigma, n)
        \\
        \sqrt{\frac{(\Lambda p)^0}{p^0}}
        \sum_{\bar{\sigma}} v_{\bar{l}}(\mathbf{p}_{\Lambda}, \bar{\sigma}, n) D_{\bar{\sigma} \sigma}^{(j_n)}(W(\Lambda, p))
        = \sum_{l} D_{\bar{l} l}(\Lambda) v_{l}(\mathbf{p}, \sigma, n)
        \end{aligned}
    \end{equation}
\end{prop}

Let us rephrase the conclusion in representation-theoretical language.
\begin{thm}$\quad$

    In order for Axiom 3 to satisfy, it is equivalent that the matrix 
    $\{u_l(\mathbf{p},\sigma,n)\}_{l;\mathbf{p},\sigma}$ is a homomorphism of representations of $SL(2,\mathbb{C})$ from $U_0|_{SL(2,\mathbb{C})}$ to $D$, 
    and $\{v_l(\mathbf{p},\sigma,n)\}_{l;\mathbf{p},\sigma}$ is a homomorphism of representations of $SL(2,\mathbb{C})$ from the dual of $U_0|_{SL(2,\mathbb{C})}$ to $D$.
\end{thm}
\begin{proof}
    Applying $\bar{l}$-component at two sides of
    \begin{equation*}
        \square U_0(\Lambda, 0) \Phi_{\mathbf{p}, \sigma} = D(\Lambda) \square \Phi_{\mathbf{p}, \sigma}
    \end{equation*}
    where $\square$ is the linear map from $\mathcal{H}$ to the action space of $D$, determined by $\{u(\mathbf{p},\sigma,n)\}_{l;\mathbf{p},\sigma}$. By \eqref{single particle state 2}, 
    \begin{equation*}
        \begin{aligned}
        U_0(\Lambda, 0) \Psi_{p, \sigma}
        = 
        \sqrt{\frac{(\Lambda p)^0}{p^0} } \sum_{\sigma'} D^{(j)}_{\sigma' \sigma}(W(\Lambda, p)) \Psi_{\Lambda p, \sigma'}
        \end{aligned}
    \end{equation*}
    $$
    LHS = \sqrt{\frac{(\Lambda p)^0}{p^0}}
    \sum_{\bar{\sigma}} 
    D_{\bar{\sigma} \sigma}^{(j_n)}(W(\Lambda, p)) u_{\bar{l}} (\mathbf{p}_{\Lambda}, \bar{\sigma}, n)
    \delta_{\bar{l}}
    $$
    $RHS = D(L(p)) (\sum_l u_l(\mathbf{p},\sigma) \delta_l) = \sum_{\bar{l}l} D_{\bar{l}l}(L(p)) u_l(\mathbf{p},\sigma) \delta_{\bar{l}}$. 
    So the $\bar{l}$-component of two sides are just the two sides of \eqref{A'}.

    The conclusion for $v$'s comes from definition of dual rep: $\rho^*(g)=\rho(g)^{-T}$. And unitarity gives $-T=- \dagger *=*$. 
\end{proof}
$\quad$
\\\\

\centerline{Step 2}

Now the Frobenius reciprocity comes into use. 

\begin{thm}$\quad$

    In order for Axiom 3 to satisfy, it is equivalent that\\
    \fbox{%
  \parbox{1\textwidth}{
    \begin{equation}\label{B}
        \begin{aligned}
            u_{\bar{l}}(\mathbf{q}, \sigma, n)=(m / q^0)^{1 / 2} \sum_{l} D_{\bar{l} l}(L(q)) u_{l}(0, \sigma, n)
            \\
            v_{\bar{l}}(\mathbf{q}, \sigma, n)=(m / q^0)^{1 / 2} \sum_{l} D_{\bar{l} l}(L(q)) v_{l}(0, \sigma, n)
        \end{aligned}
    \end{equation}
    }
    }
    \\\\
    and:\\
    \fbox{%
      \parbox{1\textwidth}{
    \begin{equation}
        \begin{aligned}\label{C}
            \sum_{\bar{\sigma}} u_{\bar{l}}(0, \bar{\sigma}, n) D_{\bar{\sigma} \sigma}^{(j_n)}(W)
            =\sum_{l} D_{\bar{l} l}(W) u_{l}(0, \sigma, n)
            \\
            \sum_{\bar{\sigma}} v_{\bar{l}}(0, \bar{\sigma}, n) D_{\bar{\sigma} \sigma}^{(j_n)^*}(W)
            =\sum_{l} D_{\bar{l}l}(W) v_{l}(0, \sigma, n)
            \\
            \forall W\in \mathcal{W}
        \end{aligned}
    \end{equation}
    }
    }
\\
\end{thm}
\begin{proof}
    Recall that Frobenius reciprocity \ref{Frobenius reciprocity} says
    \begin{equation}
     \begin{aligned}
         Hom_{\mathcal{W}}(\mathcal{H}_k, Res D) &\cong Hom_{SL(2,\mathbb{C})} (\mathcal{H}, D)\\
         F &\mapsto (g_iv \mapsto D(g_i)F(v))\\
         G|_V &\leftarrow G
     \end{aligned}
     \end{equation}
    Let me translate it into formulas. 
    In the right arrow, 
    $$
    F: \Psi_{0,\sigma} \mapsto \sum_{\bar{l}} u_{\bar{l}}(0,\sigma) \delta_{\bar{l}}
    $$
    and thus
    $$
    \tilde{F}: \Psi_{\mathbf{p}, \sigma} \mapsto \sqrt{\frac{m}{p^0}} D(L(p)) \sum_l u_l(0,\sigma) \delta_l
    =
    \sqrt{\frac{m}{p^0}} \sum_{\bar{l}l} D_{\bar{l}l}(L(p)) u_l(0,\sigma) \delta_{\bar{l}}
    $$
    On the other side, clearly $\tilde{F}(\Psi_{\mathbf{p},\sigma}) = \sum_{\bar{l}} u_{\bar{l}}(\mathbf{p},\sigma)\delta_{\bar{l}}$.
    
    The conclusion for $v$'s is similar. 
    
    So evaluating at $\delta_{\bar{l}}$, the right arrow gives us the first. 
    The left arrow gives us the second. 
\end{proof}

\begin{rmk}
    The dual rep of a group rep $D$ is $\Lambda\mapsto D(\Lambda^{-1})^T$. In this case, it equals its complex conjugate because state rep is unitary. 
\end{rmk}

%%%%%%%%%%%%%%%%%%%%%%%%%%%%%%%%%%%%%%%%%%%%%%%%%%%%%%%%%%%%%%%%%%%%%%

We can transform it into the more convenient Lie algebra language. 

\begin{thm}
\footnote{I have never seen the conclusion that homomorphism between Lie group representations is equivalent to homomorphism between Lie algebra representations in any literature. This comes from conversation with my classmate Hao Zhang. }
\label{Main-Relation between two representations}

    \eqref{C} is equivalent that the matrix $(u_l(0,\sigma,n))_{l;\sigma}$ is a homomorphism of representations of Lie algebra $w$ of the little group $\mathcal{W}$, from $dD^{(j_n)}$ to $dD|_{w}$,
    and the matrix $(v_l(0,\sigma,n))_{l;\sigma}$ is a homomorphism of representations from the dual of $dD^{(j_n)}$ to $dD|_{w}$. 
\end{thm}
\begin{proof}
    $\Longrightarrow$: clear. 

    $\Longleftarrow$: As is well know by \cite{warner1983foundations} that:
    \begin{itemize}
        \item exp is a diffeomorphism near $0$ of its Lie algebra
        \item Any open neighborhood of identity generate a subgroup that is the connected component of the identity. 
        \item
        exp commutes with Lie group-Lie algebra morphism. 
    \end{itemize}
    so for any connected Lie group, all of its elements can be written as a finite product of exponentials of its Lie algebra element: $h=e^{X_1}\cdots e^{X_m}$. For all $X\in w$, $T(\rho_1(X)^kv)=\rho_2(X)^k T(v)$. So $T(e^{\rho_1(X)} v)=e^{\rho_2(X)} T(v)$. So by induction. 

    The conclusion for $v$'s is similar. 
\end{proof}

\begin{rmk}
    The dual rep for a rep of Lie algebra is $\pi^\vee(X)=-\pi(X)^T$. And we have $-\pi(J)^T=-\pi(J)^{\dagger T}=-\pi(J)^*$ in Hermitian case (and Hermitian is by physical basis of Lie algebra, which adds an $i$ factor).
\end{rmk}

Note that the result until this subsection applies to both massive and massless case. 

\clearpage

\subsection{Massive field--relation between spins of state and field}

Now assuming positive mass. 

%%%%%%%%%%%%%%%%%%%%%%%%%%%%%%%%%%%%%%%
Now make use of the results of finite dimensional representation theory of Lorentz algebra, suppose the field rep is $\oplus_{i=1}^n \pi_{A_i}\boxtimes \pi_{B_i}$. 
\begin{prop}$\quad$

    The matrix $(u_{l}(0,\sigma))_{l,\sigma}$
    is a homomorphism of representations of $su(2)$ from $\pi_j$ to $\oplus_{i=1}^n \pi_{A_i}\otimes \pi_{B_i}$, in standard basis. 
\end{prop}
\begin{proof}
    As is remarked in section 1.3, 
    $$
    Res^{so(1,3)_\mathbb{C}}_{span\mathbf{J}} (\oplus_{i=1}^n \pi_{A_i}\boxtimes \pi_{B_i}  )
    = \oplus_{i=1}^n \pi_{A_i} \otimes \pi_{B_i}
    $$
    as representations of $su(2)$. And the standard basis corresponds to each other. 
\end{proof}
For the solving of $v$'s, let's notice that
\begin{lemma}$\quad$

We have
    \begin{equation}
        - \mathbf{J}_{\bar{\sigma} \sigma}^{(j) *} = (-1)^{\sigma-\bar{\sigma}}\mathbf{J}_{-\bar{\sigma}, -\sigma}^{(j)}
    \end{equation}
\end{lemma}
\begin{rmk}
    This is to say that $\pi_J$ is self-dual: 
    \begin{equation}
        - \mathbf{J}^{(j) *} = C \mathbf{J}^{(j)} C^{-1}
    \end{equation}
    where $C_{\sigma,\bar{\sigma}}=(-1)^{j+\sigma} \delta_{\sigma+\bar{\sigma}}$ will be part of the C of CPT.
\end{rmk}

So we can choose (though it is not always the case)
\begin{equation}
    v_{l}(0,\sigma) = ( u C^{-1} )_{l,\sigma} = (-1)^{-j-\sigma} u_{l}(0,-\sigma) = (-1)^{j+\sigma} u_{l}(0,-\sigma)
\end{equation}

By \eqref{A} and \eqref{B}, we see that in order for the quantum field not equal to zero it is equivalent to $u_l(0,\sigma)$ not all zero. This is by the 'angular momentum coupling', which reads

\begin{thm}(Relation between spin of state and field)\label{Main-Spin of states and fields}

    In order for the the existence of a non-zero massive quantum field with spin 
    $$
    \bigoplus_{i=1}^n (A_i,B_i)
    $$ 
    describing particles with spin $j$ satisfying Axiom 3, it is equivalent to the condition
    \begin{equation}
        j \in \{ |A_i-B_i|, \cdots, A_i+B_i-1,  A_i+B_i \quad | \quad i=1,\cdots,n \}
    \end{equation}
\end{thm}

This is the main result of this subsection.

\begin{eg}$\quad$

    Scalar field is defined as the field where the field representation $D$ is the trivial rep of the Lorentz group. So massive scalar field describes particles that have spin 0. 

    Vector field has field rep $(\frac{1}{2}, \frac{1}{2})$, or $D(\Lambda)=\Lambda$. So massive vector field describes particles with spin 0 or 1. 

    Dirac field has field rep $(\frac{1}{2}, 0) \oplus (0, \frac{1}{2})$. So massive Dirac field describes particles with spin $\frac{1}{2}$. 
\end{eg}

Actually we can explicitly write down the solution when the field rep is irreducible, which is the case for scalar and vector field, but not for Dirac field, which solution relies on CPT. 

Now suppose the field representation is $\pi_{(A,B)}$. The corresponding subscripts are denoted 
$a=-A,\cdots, A, b = -B,\cdots, B$.

Obviously the zero-momentum coefficient functions $u_{a b}(0, \sigma)$ are just the CG-coefficients! So we can choose a normalization, and the solution to zero-momentum coefficients are:
\begin{prop}

    We have a solution to \eqref{C}:
    \begin{equation}
        \begin{aligned}
            u_{ab}(0,\sigma) &= (2m)^{-1/2} \langle AaBb| j\sigma \rangle
            = (2m)^{-1/2} C_{A B}(j \sigma ; ab)
            \\
            v_{ab}(0,\sigma) &= (-1)^{j+\sigma} u_{ab}(0,-\sigma)
        \end{aligned}
    \end{equation}
    And the two fields $\psi^{\pm}$ are unique up to two scalars. 
\end{prop}

\clearpage

%%%%%%%%%%%%%%%%%%%%%%%%%%%%%%%%%%%%%%%%%%%%%%%%%%%%%%%%%%%%%%%%%%%%%%

\subsection{Massless field--relation between helicity and spin of field}

Recall that the zero momentum coefficients gather as a homomorphism from state representation to field representation of the Lie algebra of the little group. We only consider irreducible field representation here. 

The Lie algebra of $\mathcal{W}$ is spanned by $J_2-K_1, -J_1-K_2, J_3$, which in state representation acts as
\begin{equation}
\begin{aligned}
\pi^\sigma(J_3) &= \sigma \\
\pi^\sigma(J_2-K_1) = \pi^\sigma(-J_1-K_2) &= 0
\end{aligned}
\end{equation}

The conditions for standard momentum coefficients are
\footnote{
\cite{weinberg2002quantum} stated that the equations for $v$ are just the complex conjugates of the equations for $u$, so we can choose
$v_{l}(\mathbf{p}, \sigma)=u_{l}(\mathbf{p}, \sigma)^*$. But it's not correct. 
}
\begin{equation}\label{C1}
\begin{aligned}
\sigma u_{a b}(\mathbf{k}, \sigma) & =(a+b) u_{a b}(\mathbf{k}, \sigma)\\
-\sigma v_{a b}(\mathbf{k}, \sigma) & =(a+b) v_{a b}(\mathbf{k}, \sigma)
\end{aligned}
\end{equation}
and
\begin{equation}\label{C2}
\begin{aligned}
0 & =(\pi_{AB}(J_2)-\pi_{AB}(K_1))_{a b, a' b'} u_{a' b'}(\mathbf{k}, \sigma) \\
& =(J_2^{(A)}+i J_1^{(A)})_{a a'} u_{a' b}(\mathbf{k}, \sigma)+(J_2^{(B)}-i J_1^{(B)})_{b b'} u_{a b'}(\mathbf{k}, \sigma)\\
& = (-i J_-^{(A)})_{a a'} u_{a' b}(\mathbf{k}, \sigma)+(-i J_+^{(B)})_{b b'} u_{a b'}(\mathbf{k}, \sigma)\\
0 & =(-\pi_{AB}(J_1)-\pi_{AB}(K_2))_{a b, a' b'} u_{a' b'}(\mathbf{k}, \sigma) \\
& =(-J_1^{(A)}+i J_2^{(A)})_{a a'} u_{a' b}(\mathbf{k}, \sigma)+(-J_1^{(B)}-i J_2^{(B)})_{b b'} u_{a b'}(\mathbf{k}, \sigma)\\
& =(-J_-^{(A)})_{a a'} u_{a' b}(\mathbf{k}, \sigma)+(-J_+^{(B)})_{b b'} u_{a b'}(\mathbf{k}, \sigma)
\end{aligned}
\end{equation}

By \eqref{C1}, $u_{a b}(\mathbf{k}, \sigma)$ and $v_{a b}(\mathbf{k}, \sigma)$ must vanish unless $\sigma=a+b$ and $\sigma=-a-b$, respectively. 
\eqref{C2} requires that $u_{a b}(\mathbf{k}, \sigma)$ vanishes unless
\begin{equation}
a=-A, b=+B
\end{equation}
and the same is obviously also true of $v_{a b}(\mathbf{k}, \sigma)$.

In summary, we get
\begin{thm}(Relation between helicity and field spin)

The existence for a non zero $(A,B)$ field corresponding to particle with helicity $\sigma$ satisfying Axiom 3 is equivalent to 
\begin{equation}
    \sigma = \pm(-A+B)
\end{equation}
where $+$ corresponds to non zero $\psi^+$ and zero $\psi^-$, and $-$ corresponds to zero $\psi^+$ and non zero $\psi^-$.
\end{thm}

This result applies to old neutrino theory (\cite{weinberg2002quantum}), where neutrino was though of a massless particle with helicity $-\frac{1}{2}$. A $(\frac{1}{2},0)$ field annilates neutrinos and creates anti-neutrinos, and a $(0,\frac{1}{2})$ field annilates anti-neutrinos and creates neutrinos. At that point we combine two fields with same field spin but with different helicity, i.e. field rep. 

At this point it is clear that we cannot construct non-zero vector field satisfying Axiom 3 to describe massless particle with helicity $\pm 1$, which is the case of photon. This is a major setback we face in QED. We can construct quantum field with spin $(1,0)\otimes (0,1)$, i.e. anti-symmetric covariant tensor $f_{\mu\nu}$. But for some reason, we would rather loosen the condition of Axiom 3, add an extra term and lead to gauge transformation. 
\footnote{See \cite{weinberg2002quantum}.}

\clearpage

% \subsection{Field equation for arbitrary spin}

\subsection{Discussions}

1.

    What's the difference between the state representation in Section 2 and the field representation in Section 3?

    We just compare the irreducible case, and for the non-irreducible case we must pass to their universal covers. See the table below. 
    \\
    \newcommand{\tabincell}[2]{\begin{tabular}{@{}#1@{}}#2\end{tabular}}  

    \begin{table}[htbp]
    \caption{Comparison of state rep and field rep}
    \begin{tabular}{cccc}% 其中，tabular是表格内容的环境；c表示centering，即文本格式居中；c的个数代表列的个数
    \toprule %[2pt]设置线宽     
        & State rep  &  Field rep & Comparison \\ %换行
    \midrule %[2pt]  
    Group & Poincare & Lorentz & \tabincell{c}{whether contain\\ non-homogeneous part}\\
    Representation space & \tabincell{c}{Hilbert space of\\ quantum states} & \tabincell{c}{Target space\\ of field} & Fundamentally different \\
    Dimension & infinite & finite &  Fundamentally different \\
    whether unitary & Must be & Must not & Fundamentally different \\
    whether projective  & possibly & possibly & Same \\
    Parameter & \tabincell{c}{ mass $m>0$, spin $j$\\(massless)helicity $\sigma$} & spin $(A,B)$&  \\
    \bottomrule %[2pt]
    \end{tabular}
    \end{table}

2.

    How to understand that state representation must be unitary, while field representation must not be unitary?

    The first is because the physical meaning of (the square of) inner product of quantum states, we must let it to be so.
    
    For the second, this is because we require the target space of field to be finite dimensional, and Lorentz group has no nontrivial finite dimensional projective representation. Heuristicly speaking this is because Lorentz group is not compact, but this is not a rigorous argument. We arrive at this argument after we derive the general representation of Lorentz algebra in Section 1.

3. 

    What's the physical meaning of the above argument?

    It's OK for field rep tp be non-unitary, because the target space of field rep is not the space of quantum states which means we do not need 'conservation of probability'. This is why Dirac equation is not understood as a wave equation as in the history, but should be understood as a field equation.

4.

    I'd like to talk about some problems in many physics literature. 

    First, there is no Poincare group's projective representation on Fock space. This is because projective representations cannot have a direct sum in general. 

    Second, the field rep is also of $SL(2,\mathbb{C})$, not $SO^+(1,3)$. For example, the Dirac rep is not a linear rep of the Lorentz group. 

    So for either convenience or correctness, all should pass to universal cover's linear representation. Many of the physics literature ignored these and make wrong assertions.

\clearpage

\section*{Acknowledgement}

The main part of this article is completed during the period from 2022.12.16 to 2023.1.2 as the final paper for a course. Some of the materials is discovered by myself, though it may not be academically innovative. 

I would like to express my gratitude to professor Boqiang Ma for his course <Astroparticle Physics> in Peking University, for leading me into the palace of Quantum Field Theory in Spring 2022. 

I would also like to thank Qianchu Yi, for her sharing with me the game <Dyson Sphere Program>, coincidentally beginning from 2022.12.17, several hours after the start of this article and several hours before I'm infected with COVID-19.

\clearpage

\nocite{*}
\bibliographystyle{alpha}
\bibliography{ref.bib}

\newcommand{\etalchar}[1]{$^{#1}$}
\begin{thebibliography}{{Wik}22b}

\bibitem[EGH{\etalchar{+}}11]{etingof2011introduction}
Pavel~I Etingof, Oleg Golberg, Sebastian Hensel, Tiankai Liu, Alex Schwendner,
  Dmitry Vaintrob, and Elena Yudovina.
\newblock {\em Introduction to representation theory}, volume~59.
\newblock American Mathematical Soc., 2011.

\bibitem[Fol16]{folland2016course}
Gerald~B Folland.
\newblock {\em A course in abstract harmonic analysis}, volume~29.
\newblock CRC press, 2016.

\bibitem[Hum12]{humphreys2012introduction}
James~E Humphreys.
\newblock {\em Introduction to Lie algebras and representation theory},
  volume~9.
\newblock Springer Science \& Business Media, 2012.

\bibitem[Poo22]{poon2022projective}
Levi Poon.
\newblock Projective representation theory.
\newblock 2022.

\bibitem[S{\etalchar{+}}77]{serre1977linear}
Jean-Pierre Serre et~al.
\newblock {\em Linear representations of finite groups}, volume~42.
\newblock Springer, 1977.

\bibitem[Sch14]{schwartz2014quantum}
Matthew~D Schwartz.
\newblock {\em Quantum field theory and the standard model}.
\newblock Cambridge University Press, 2014.

\bibitem[Str08]{straumann2008unitary}
Norbert Straumann.
\newblock Unitary representations of the inhomogeneous lorentz group and their
  significance in quantum physics.
\newblock {\em arXiv preprint arXiv:0809.4942}, 2008.

\bibitem[TAY]{taylorrepresentations}
JAY TAYLOR.
\newblock Representations of sl2 (c).

\bibitem[War83]{warner1983foundations}
Frank~W Warner.
\newblock {\em Foundations of differentiable manifolds and Lie groups},
  volume~94.
\newblock Springer Science \& Business Media, 1983.

\bibitem[Wei02]{weinberg2002quantum}
Steven Weinberg.
\newblock {\em The quantum theory of fields: Foundations}.
\newblock Cambridge University Press, 2002.

\bibitem[Wig39]{wigner1939unitary}
Eugene Wigner.
\newblock On unitary representations of the inhomogeneous lorentz group.
\newblock {\em Annals of mathematics}, pages 149--204, 1939.

\bibitem[{Wik}22a]{enwiki:1079823937}
{Wikipedia contributors}.
\newblock Projective representation --- {Wikipedia}{,} the free encyclopedia.
\newblock
  \url{https://en.wikipedia.org/w/index.php?title=Projective_representation&oldid=1079823937},
  2022.
\newblock [Online; accessed 31-December-2022].

\bibitem[{Wik}22b]{enwiki:1123050261}
{Wikipedia contributors}.
\newblock Representation theory of the lorentz group --- {Wikipedia}{,} the
  free encyclopedia.
\newblock
  \url{https://en.wikipedia.org/w/index.php?title=Representation_theory_of_the_Lorentz_group&oldid=1123050261},
  2022.
\newblock [Online; accessed 31-December-2022].

\end{thebibliography}

\end{document}